\documentclass[11pt,a4paper]{article}

\usepackage[utf8]{inputenc}
\usepackage[T1]{fontenc}
\usepackage{amsmath,amssymb,amsthm}
\usepackage[margin=2.5cm]{geometry}
\usepackage{enumitem}
\usepackage{graphicx}
\usepackage[colorlinks=true,linkcolor=blue,citecolor=blue,urlcolor=blue]{hyperref}
\usepackage{booktabs}

\theoremstyle{plain}
\newtheorem{theorem}{Theorem}[section]
\newtheorem{proposition}[theorem]{Proposition}
\newtheorem{lemma}[theorem]{Lemma}
\newtheorem{corollary}[theorem]{Corollary}

\theoremstyle{definition}
\newtheorem{definition}[theorem]{Definition}
\newtheorem{principle}[theorem]{Modelling Principle}

\theoremstyle{remark}
\newtheorem{remark}[theorem]{Remark}

\newcommand{\R}{\mathbb{R}}
\newcommand{\Pos}{\mathrm{Pos}}

\begin{document}

\begin{center}
{\Large\bfseries Special Markowitz:\hspace{-.2em} Thermodynamic Formalism for the\\[6pt]
Joint Regularisation of Returns and Covariance}

\bigskip\bigskip

David Reinhardt\\[6pt]
{\itshape studio\hspace{-.15em} entropica, Sch\"aftlarn/Isartal, Germany\\[2pt]
dr.david.reinhardt@entropica.studio}
\end{center}

\bigskip

\begin{abstract}
Special Markowitz (SM) regularises returns and covariance jointly,
relative to a reference state $(\mu_{\mathrm{ref}},\Sigma_{\mathrm{ref}})$.
Each eigendirection of the whitened relative operator carries a
signed spectral potential
$\Phi_k\in\R$, with persistence factor $\psi_k=e^{-\Phi_k}>0$.
Positive potentials attenuate empirical deviations from the
reference geometry, zero potential preserves them, and negative
potentials amplify them.
The persistence factor~$\psi_k$ governs both the return signal
and the covariance deviation: the regularised deviation from
the reference is~$\psi_k$ times the empirical deviation.
The logarithmic potential coordinate is characterised by a
multiplicative composition law on the group $((0,\infty),\times)
\cong(\R,+)$;
the Stein loss is characterised
as the unique free-energy density (within a natural class)
compatible with the resulting coupling.
The SM pressure functional is additive across modes --- the
defining property of Special Markowitz.
\end{abstract}

\section{Working Definition}

Special Markowitz applies a signed spectral potential field to
jointly regularise return and covariance
relative to a reference state $(\mu_{\mathrm{ref}},\Sigma_{\mathrm{ref}})$.

\smallskip\noindent
This determines:
\begin{itemize}[nosep,leftmargin=2em]
\item the \textbf{space}: $\Pos(n)\times\R^n$ --- covariance
  operators and returns;
\item the \textbf{reference state}:
  $(\mu_{\mathrm{ref}},\Sigma_{\mathrm{ref}})$ --- the equilibrium,
  the prior;
\item the \textbf{empirical state}: the data $(\hat\mu,\hat\Sigma)$ ---
  deforming the equilibrium;
\item the \textbf{potential field}:
  $(\Phi_1,\ldots,\Phi_n)\in\R^n$ --- controlling the per-mode
  coupling to the reference;
\item the \textbf{variational principle}: minimise the modal free
  energy, with potential-dependent reference coupling;
\item the \textbf{result}: an equilibrium Gaussian state --- the
  regularised distribution $(\mu^*,\Sigma^*)$.
\end{itemize}

\smallskip\noindent
We use the terminology of thermodynamic formalism in the
mathematical sense of Ruelle: potentials, exponential weights,
variational functionals, equilibrium states, and pressure.
Whether the SM pressure admits an exact representation as a Ruelle
pressure, with an associated invariant measure and entropy
functional, is not established here.

\section{Thermodynamic Motivation}
\label{sec:motivation}

This section traces the conceptual path that leads
to the SM construction, starting from a classical scalar regularisation problem and
arriving, by three natural steps, at a spectral free-energy
formulation.

\subsection{Scalar regularisation}

Consider the standard Gaussian shrinkage problem: given an
empirical state~$\hat{\mathcal{N}} = \mathcal{N}(\hat\mu,\hat\Sigma)$
and a reference state~$\mathcal{N}_{\mathrm{ref}} =
\mathcal{N}(\mu_{\mathrm{ref}},\Sigma_{\mathrm{ref}})$, find a
regularised state~$\mathcal{N}^* = \mathcal{N}(\mu^*,\Sigma^*)$
that balances fidelity to the data against proximity to the reference:
\begin{equation}\label{eq:mot-scalar-KL}
(\mu^*,\Sigma^*)
= \operatorname*{arg\,min}_{(\mu,\Sigma)}
\Bigl\{
  D_{\mathrm{KL}}\!\bigl(\mathcal{N}(\mu,\Sigma)
    \,\big\|\,\hat{\mathcal{N}}\bigr)
  \;+\;
  \eta\,
  D_{\mathrm{KL}}\!\bigl(\mathcal{N}(\mu,\Sigma)
    \,\big\|\,\mathcal{N}_{\mathrm{ref}}\bigr)
\Bigr\},
\end{equation}
where $\eta \ge 0$ is a single scalar coupling parameter.
When both KL terms share the same informational geometry, the
first-order condition yields a barycentric solution of the
schematic form
\begin{equation}\label{eq:mot-bary}
\mathcal{N}^*
= \frac{1}{1+\eta}\,\hat{\mathcal{N}}
  \;+\;\frac{\eta}{1+\eta}\,\mathcal{N}_{\mathrm{ref}}\,,
\end{equation}
where the ``$+$'' denotes the appropriate mixture in the chosen
geometry (not a na\"ive linear combination of measures).
For Gaussian states this barycentric notation refers to the
corresponding information-geometric interpolation in natural
parameters; it does not imply componentwise linear interpolation
of $\mu$ and~$\Sigma$.
The fraction $1/(1+\eta)$ controls how much of the empirical signal
persists; the complement $\eta/(1+\eta)$ is the shrinkage toward
the reference.

\subsection{From coupling to potential}

Define
\begin{equation}\label{eq:mot-Phi}
\boxed{\;\Phi := \ln(1+\eta)\;}\,,
\qquad
\text{equivalently}\quad
\eta = e^{\Phi}-1\,.
\end{equation}
Then the persistence and shrinkage fractions become exponential
weights,
\begin{equation}\label{eq:mot-persistence}
\underbrace{\frac{1}{1+\eta}}_{\text{persistence}}
= e^{-\Phi},
\qquad\qquad
\underbrace{\frac{\eta}{1+\eta}}_{\text{shrinkage}}
= 1 - e^{-\Phi},
\end{equation}
and the barycentric solution~\eqref{eq:mot-bary} takes the form
\begin{equation}\label{eq:mot-bary-Phi}
\boxed{\;
\mathcal{N}^*
= e^{-\Phi}\,\hat{\mathcal{N}}
  \;+\;(1 - e^{-\Phi})\,\mathcal{N}_{\mathrm{ref}}
\;}\,.
\end{equation}

\medskip\noindent
\textbf{Why the logarithm?}\enspace
The reparametrisation $\eta \mapsto \Phi$ is not merely cosmetic.
It is singled out by the composition structure of independent
persistence filters.

If two independent filters with persistence $\psi_A$ and~$\psi_B$
are applied sequentially, their combined persistence is
multiplicative:
$\psi_{A \circ B} = \psi_A\,\psi_B$.
Requiring an \emph{additive} coordinate for this composition
forces, up to a positive scale factor, the logarithmic
parametrisation
\begin{equation}\label{eq:mot-log-forced}
\Phi = -\ln\psi,
\qquad
\Phi_{A \circ B} = \Phi_A + \Phi_B\,.
\end{equation}
This additive structure is precisely the one naturally
accommodated by thermodynamic formalism, where potentials enter
\emph{linearly} in the variational principle.  The classical
pressure functional (Ruelle~[5]) has the form
\begin{equation}\label{eq:mot-pressure}
P(\varphi)
= \sup_{\nu}\;\bigl[\,\langle\varphi,\nu\rangle
  + H(\nu)\,\bigr],
\end{equation}
where $\varphi$~is a potential, $\nu$~ranges over admissible
states, and $H(\nu)$ is an entropy-like functional.  The
logarithmic parametrisation does not follow \emph{from} Ruelle's
variational principle; rather, multiplicative persistence produces
additive log-potentials, and Ruelle's formalism provides the
natural variational language for potentials with exactly this
additive structure.

\begin{remark}[Parameter chain]
\label{rem:mot-chain}
The full parameter chain is
\[
\eta
\;\longleftrightarrow\;
\alpha = \frac{\eta}{1+\eta}
\;\longleftrightarrow\;
\psi = 1 - \alpha = e^{-\Phi}
\;\longleftrightarrow\;
\Phi = -\ln\psi\,.
\]
Among these, $\eta$~is the Lagrange coupling, $\alpha$~the
shrinkage intensity, $\psi$~the persistence factor, and $\Phi$~the
thermodynamically natural coordinate.  The algebraic parameter
ranges are
\[
  \eta\in(-1,\infty),\qquad
  \alpha\in(-\infty,1),\qquad
  \psi\in(0,\infty),\qquad
  \Phi\in\R.
\]
Mode-dependent admissibility (Definition~\ref{def:admissible})
restricts these ranges further when~$\lambda_k<1$.
The distinguished point is
$\eta=0\;\Leftrightarrow\;\alpha=0
\;\Leftrightarrow\;\psi=1
\;\Leftrightarrow\;\Phi=0$:
the identity transformation that preserves the empirical state.
\end{remark}

\subsection{The spectral potential field}

Equation~\eqref{eq:mot-bary-Phi} reveals the limitation of the scalar
formulation: a single~$\Phi$ forces \emph{all spectral directions}
to share the same persistence~$e^{-\Phi}$.  In practice, different
eigenmodes of the covariance carry very different amounts of
statistical evidence.  The natural generalisation is therefore
\begin{equation}\label{eq:mot-Phi-field}
\Phi \;\longrightarrow\; (\Phi_1,\ldots,\Phi_n)\,,
\qquad
e^{-\Phi} \;\longrightarrow\; e^{-\Phi_k}\,,
\end{equation}
a \emph{spectral potential field} that assigns each eigenmode its
own spectral potential.

\medskip\noindent
\textbf{Which eigenmodes?}\enspace
The coordinates~$k$ are not the original asset directions.  They
are determined by the \emph{relative geometry} of empirical and
reference covariance.

The guiding idea comes from Bianconi's \emph{Gravity from
Entropy}~[7], where two metrics $G$ (matter-induced)
and~$g$ (reference/geometric) are compared via the action
\begin{equation}\label{eq:mot-GfE}
\mathcal{L}
= -\operatorname{Tr}\ln\!\bigl(G\,g^{-1}\bigr)
= -\sum_{k=1}^n \ln\lambda_k\,,
\end{equation}
where $\lambda_1,\ldots,\lambda_n$ are the eigenvalues of the
relative operator~$G\,g^{-1}$.  This action is manifestly
basis-independent; the eigenbasis of $G\,g^{-1}$ is the natural
coordinate system in which the two metrics are compared.

Transferring this to the portfolio setting, let $G \to \hat\Sigma$
(empirical covariance) and $g \to \Sigma_{\mathrm{ref}}$
(reference covariance).  The generally non-symmetric relative
operator~$\hat\Sigma\,\Sigma_{\mathrm{ref}}^{-1}$ has the same
spectrum as its symmetric representative
$A := \Sigma_{\mathrm{ref}}^{-1/2}\,
   \hat\Sigma\,
   \Sigma_{\mathrm{ref}}^{-1/2}$
(defined formally in Section~\ref{sec:setup}).
Bianconi's action evaluated on this pair gives
$\mathcal{L}_{\mathrm{GfE}}
= -\operatorname{Tr}\ln A
= -\sum_{k} \ln\lambda_k$.

The Stein divergence of the same pair decomposes as
\begin{equation}\label{eq:mot-Stein-decomp}
D_{\mathrm{Stein}}\!\bigl(\hat\Sigma \,\big\|\, \Sigma_{\mathrm{ref}}\bigr)
= \sum_{k} (\lambda_k - \ln\lambda_k - 1)
= \underbrace{\textstyle\sum_k \lambda_k}_{\operatorname{Tr} A}
\;+\;\underbrace{\bigl(-\textstyle\sum_k \ln\lambda_k\bigr)}_
    {\mathcal{L}_{\mathrm{GfE}}}
\;-\; n\,.
\end{equation}
Under the portfolio identification $G\mapsto\hat\Sigma$,
$g\mapsto\Sigma_{\mathrm{ref}}$, the logarithmic part of
Bianconi's relative-metric action coincides algebraically with
the log-determinant term of the Stein divergence.  Adding the
trace term~$\operatorname{Tr} A = \sum_k \lambda_k$ yields the
familiar energy-minus-log-volume structure of the Stein loss.
This suggests a thermodynamic reading:
$D_{\mathrm{Stein}}
= \underbrace{\operatorname{Tr} A}_{\text{``energy''}}
  +\underbrace{\mathcal{L}_{\mathrm{GfE}}}_{\text{``neg.\ entropy''}}
  - n$,
i.e.\ the Stein loss as a free-energy functional whose entropic
component is Bianconi's geometric action.

\begin{remark}[Two routes to the Stein loss]
\label{rem:mot-two-routes}
The Stein loss appears here from two independent directions:
(i)~the logarithmic component of Bianconi's geometric
action suggests a free-energy reading whose completion is
$D_{\mathrm{Stein}}$, and (ii)~later, the SM characterisation
theorem (Proposition~\ref{prop:characterisation}) identifies
$D_{\mathrm{Stein}}$ as the \emph{unique} covariance divergence
compatible with the Coupling Identity within the assumed loss
class.  That both routes single out the same divergence is a
non-trivial consistency.
\end{remark}

\subsection{The scalar KL does not extend to a spectral potential}

A natural first attempt would be to replace $\eta$ in the
scalar problem~\eqref{eq:mot-scalar-KL} by mode-dependent
couplings~$\eta_k$.  But the Gaussian KL divergence is a
\emph{global} functional of the full distribution; it does not
decompose into independent modal contributions with freely
adjustable weights.  Writing
``$\sum_k \eta_k\,D_{\mathrm{KL}}^{(k)}$'' is not canonically
induced by the original global Gaussian KL; defining such
mode-dependent terms requires additional spectral structure.

This is precisely the point at which the architecture of
thermodynamic formalism becomes relevant.  In Ruelle's variational
principle~\eqref{eq:mot-pressure}, a \emph{non-constant} potential is
part of the basic setup; no assumption of spatial homogeneity is
needed.  The transition
\[
\underset{\text{scalar KL coupling}}{\eta}
\;\longrightarrow\;
\underset{\text{potential coordinate}}{\Phi = \ln(1+\eta)}
\;\longrightarrow\;
\underset{\text{spectral potential field}}{(\Phi_1,\ldots,\Phi_n)}
\]
therefore calls for a new, modally defined variational principle
rather than a patched version of the global Gaussian~KL.

\subsection{Spectral free energy}

With the SM sign convention, the Gibbs weight is~$e^{-\Phi}$,
so the corresponding Ruelle potential is $\varphi = -\Phi$.  The pressure
formulation~\eqref{eq:mot-pressure} therefore has the dual form
\begin{equation}\label{eq:mot-dual}
-P_R(-\Phi)
= \inf_{\nu}\;F_\Phi(\nu),
\qquad
F_\Phi(\nu) := \langle\Phi,\nu\rangle - H(\nu).
\end{equation}
We write $P_R$ for the Ruelle pressure to distinguish it from the
SM pressure $P_{\mathrm{SM}}(\Phi):=-F^*(\Phi)$ defined later
(Definition~\ref{def:pressure}).  The equilibrium state minimises
the free energy; for Special Markowitz, this free-energy side is
the natural formulation, because SM is from the outset a
minimisation problem.

\emph{SM does not identify its spectral free energy with the
Ruelle pressure of an underlying financial dynamical system.
Rather, it adopts the same variational architecture: an additive
potential field coupled to state variables, with equilibrium
defined by minimisation of a free-energy functional.  Whether
$P_{\mathrm{SM}} = P_R$ for some underlying dynamical system
is not established here.}

In the spectral setting, the state is described by the modal
regularisation variables $(\sigma_k, m_k)_{k=1}^n$, the potential
field is~$(\Phi_k)_{k=1}^n$, and the free energy takes the form
\begin{equation}\label{eq:mot-SM-F}
\boxed{\;
F\bigl(\sigma, m;\,\Phi\bigr)
= \sum_{k=1}^n F_k\!\bigl(\sigma_k, m_k;\,\Phi_k\bigr)
\;}\,.
\end{equation}
\textbf{The separability
$F = \sum_k F_k$ is the defining structural assumption of Special
Markowitz.}  It means that modes are coupled to the potential field
but not to each other.  Under the modal product structure of the
admissible state space, this produces an additive pressure
$P = \sum_k P_k(\Phi_k)$ with vanishing cross-susceptibilities
(Definition~\ref{def:SM}).
General Markowitz, by contrast, would allow non-vanishing cross
terms---the analogue of interactions and curvature, and, in an
appropriate large-system limit, potentially phase-transition-like
collective behaviour.

\subsection{From Common Persistence to the Coupling Identity}

The Coupling Identity
(Theorem~\ref{thm:coupling-identity}) does not follow from the
scalar formulation alone.  In SM, common persistence of covariance
deformation is imposed as a modelling principle
(Principle~\ref{pr:common-persistence}), while the return
scaling follows from the quadratic loss structure.
Proposition~\ref{prop:characterisation} then shows that, within the
stated loss class, Stein's loss is the unique covariance divergence
capable of realising the prescribed persistence under
scale-invariant coupling.  Theorem~\ref{thm:coupling-identity}
establishes that the resulting variational equilibrium scales
both return and covariance deformation by the same persistence
factor~$\psi_k=e^{-\Phi_k}$.

\subsection{Summary}

Three independent ideas converge in the SM construction:

\begin{center}
\renewcommand{\arraystretch}{1.5}
\begin{tabular}{@{}p{3.5cm}p{5.5cm}p{4.2cm}@{}}
\toprule
\textbf{Source} & \textbf{Contribution} & \textbf{SM role} \\
\midrule
KL regularisation
  & scalar coupling $\eta$
  & starting point \\
Bianconi,\newline \emph{Gravity from Entropy}
  & $\mathcal{L} = -\mathrm{Tr}\ln(G\,g^{-1})$:\newline
    eigenmodes of\newline
    $A = \Sigma_{\mathrm{ref}}^{-1/2}
    \hat\Sigma\,\Sigma_{\mathrm{ref}}^{-1/2}$
  & natural coordinates~$k$;\newline
    $\mathcal{L}_{\mathrm{GfE}}$ = log.\ component\newline
    of Stein loss \\
Thermodynamic formalism\newline (Ruelle)
  & non-constant potential\newline
    in a variational principle
  & spectral potential~$\Phi_k$,\newline
    free energy $F = \sum F_k$ \\
\bottomrule
\end{tabular}
\end{center}

\noindent
The title \emph{Thermodynamic Formalism} refers to this
architecture: not the claim that financial markets are
thermodynamic systems, but the transfer of a mathematical
organisational principle---\emph{potential\/ $+$ entropy\/ $+$
variational principle\/ $\to$ equilibrium}---that has proved
powerful in other domains (Bowen~[6],
Ruelle~[5]).

\begin{quote}
\emph{We transfer the architecture, not the physics.}
\end{quote}

\noindent
The freedom in the choice of~$\Phi$ is a feature, not an
indeterminacy.  Different calibration models or spectral
regularisation principles---Marchenko--Pastur thresholding,
bootstrap stability, Bayesian uncertainty, or the agnostic
symmetry underlying Agnostic Risk
Parity~(Section~\ref{sec:arp})---induce different spectral
potentials; the SM formalism provides a common variational
framework in which each such choice determines the joint
regularisation of returns and covariance through the Coupling
Identity~(Theorem~\ref{thm:coupling-identity}).

\section{Setup}
\label{sec:setup}

\begin{definition}[Whitened relative operator]
\label{def:whitened}
The whitened relative operator is the symmetric positive-definite
matrix
\begin{equation}\label{eq:A}
  A := \Sigma_{\mathrm{ref}}^{-1/2}\,
       \hat\Sigma\,
       \Sigma_{\mathrm{ref}}^{-1/2}
  \;\in\; \Pos(n).
\end{equation}
Since $A$ is symmetric, it admits an orthogonal eigendecomposition
$A=U\Lambda U^\top$ with
$\Lambda=\mathrm{diag}(\lambda_1,\ldots,\lambda_n)$, $\lambda_k>0$.
The eigenvectors $\{u_k\}$ define the natural basis of the problem.
\end{definition}

\begin{remark}[Why whitening is necessary]
The asymmetric product $R=\Sigma_{\mathrm{ref}}^{-1}\hat\Sigma$
has the same eigenvalues as~$A$, but is in general not symmetric
and does not admit an orthogonal eigendecomposition. All subsequent
spectral arguments require the symmetric form~$A$.
\end{remark}

\begin{definition}[Whitened coordinates]
\label{def:whitened-coords}
Define the whitening map $z=\Sigma_{\mathrm{ref}}^{-1/2}x$.
In whitened coordinates:
\[
  \hat\Sigma_w = A = U\Lambda U^\top, \quad
  \Sigma_{\mathrm{ref},w} = I, \quad
  \hat\mu_w = \Sigma_{\mathrm{ref}}^{-1/2}\hat\mu, \quad
  \mu_{\mathrm{ref},w} = \Sigma_{\mathrm{ref}}^{-1/2}\mu_{\mathrm{ref}}.
\]
In the eigenbasis $\{u_k\}$ of~$A$, the whitened data covariance
becomes $\Lambda=\mathrm{diag}(\lambda_1,\ldots,\lambda_n)$ and the
reference covariance remains~$I$. The whitened return difference
$\delta\mu_w := \hat\mu_w - \mu_{\mathrm{ref},w}$ has components
$\delta_k = u_k^\top\delta\mu_w$.
\end{definition}

\begin{remark}[Spectral submanifold]
SM regularises spectrally: it preserves the eigendirections of~$A$
and deforms only eigenvalues and return components. The optimisation
lives on
$M_A=\{U\,\mathrm{diag}(\sigma_1,\ldots,\sigma_n)\,U^\top:\sigma_k>0\}$.
\end{remark}

\section{Signed Spectral Potential}

\begin{definition}[Signed spectral potential and persistence factor]
\label{def:signed-potential}
Each relative eigendirection~$k$ is assigned a
\textbf{signed spectral potential}
$\Phi_k\in\R$, with \textbf{persistence factor}
\begin{equation}\label{eq:Phi-def}
  \psi_k \;:=\; e^{-\Phi_k} \;>\; 0.
\end{equation}
Equivalently, $\Phi_k = -\ln\psi_k$.
The associated \textbf{reference-coupling strength} is
\begin{equation}\label{eq:coupling}
  s_k \;:=\; \psi_k^{-1}-1 \;=\; e^{\Phi_k}-1.
\end{equation}
Three regimes arise:
\[
\renewcommand{\arraystretch}{1.3}
\begin{array}{c|c|c|c}
  \Phi_k & \psi_k & s_k & \text{effect on }\lambda_k-1 \\\hline
  >0 & <1 & >0 & \text{attenuation} \\
  =0 & =1 & =0 & \text{preservation (identity)} \\
  <0 & >1 & <0 & \text{amplification}
\end{array}
\]
$\Phi_k=0$ is not a boundary; it is the neutral element.
\end{definition}

\begin{definition}[Admissible potential field]
\label{def:admissible}
A potential field $(\Phi_1,\ldots,\Phi_n)\in\R^n$ is
\textbf{admissible} for the data eigenvalues
$(\lambda_1,\ldots,\lambda_n)$ if
\begin{equation}\label{eq:admissible}
  \boxed{\;
    \lambda_k^* \;=\; 1+e^{-\Phi_k}(\lambda_k-1) \;>\; 0
  \;}
  \qquad\text{for every } k.
\end{equation}
\end{definition}

\begin{corollary}[Admissible domain]
\label{cor:admissible-domain}
The set of admissible potential fields is the open convex product
\begin{equation}\label{eq:admissible-domain}
  \mathcal{A}(\lambda)
  \;=\;
  \prod_{k:\,\lambda_k\ge 1}\R
  \;\times\;
  \prod_{k:\,\lambda_k<1}
  \bigl(\,\ln(1-\lambda_k),\;\infty\bigr).
\end{equation}
Since $\ln(1-\lambda_k)<0$ for $\lambda_k<1$, the origin
$\Phi=0$ lies in the interior:
$0\in\operatorname{int}\mathcal{A}(\lambda)$.
\end{corollary}

\begin{proof}
The condition $\lambda_k^*>0$ is equivalent to
$e^{\Phi_k}>1-\lambda_k$.
For $\lambda_k\ge 1$ the right-hand side is non-positive, so the
inequality holds for all $\Phi_k\in\R$.  For $\lambda_k<1$ it
gives $\Phi_k>\ln(1-\lambda_k)$, and
$\ln(1-\lambda_k)<0$.
\end{proof}

\begin{remark}[Equivalence of admissibility conditions]
\label{rem:admissible-equiv}
The following conditions are equivalent:
\[
  \lambda_k^*>0
  \quad\Longleftrightarrow\quad
  F_k \text{ is strictly convex in the covariance mode}
  \quad\Longleftrightarrow\quad
  \lambda_k+s_k>0.
\]
The first is geometric (positive regularised eigenvalue), the second
variational (well-posed optimisation), the third algebraic.  They
express a single constraint in three languages.
\end{remark}

\begin{remark}[Primitive and coordinate]
\label{rem:primitive}
The persistence factor~$\psi_k$ is the primitive quantity;
$\Phi_k=-\ln\psi_k$ is its natural logarithmic coordinate.
The exponential form $e^{-\Phi_k}$ arises as a consequence of this
coordinate choice, not from a physical Boltzmann distribution.
\end{remark}

\begin{proposition}[Characterisation of the logarithmic coordinate]
\label{prop:log-char}
Assume that persistence factors from independently composed sources
compose multiplicatively,
\begin{equation}\label{eq:mult-composition}
  \psi_{AB}=\psi_A\,\psi_B
  \qquad\text{(composition law),}
\end{equation}
and require that the potential coordinate $\Phi=g(\psi)$ compose
additively, $\Phi_{AB}=\Phi_A+\Phi_B$.
Then for any continuous, strictly monotone
$g\colon(0,\infty)\to\R$ with $g(1)=0$, the functional
equation $g(\psi_A\psi_B)=g(\psi_A)+g(\psi_B)$ has the unique
solution
\begin{equation}\label{eq:log-characterised}
  \Phi(\psi)=-c\ln\psi, \qquad c>0.
\end{equation}
With the normalisation $c=1$: $\Phi=-\ln\psi$.
Equivalently, $\Phi$ is the unique continuous isomorphism
$\bigl((0,\infty),\times\bigr)
\to\bigl(\R,+\bigr)$
mapping $\psi=1$ to $\Phi=0$,
up to positive scale.
\end{proposition}

\begin{proof}
Set $\psi=e^{-t}$, $t\in\R$.
Then $h(t):=g(e^{-t})$ satisfies $h(t_1+t_2)=h(t_1)+h(t_2)$
with $h$ continuous and $h(0)=0$. By the Cauchy functional equation,
$h(t)=ct$ for some $c>0$ (strict monotonicity forces $c>0$).
Hence $g(\psi)=-c\ln\psi$.
\end{proof}

\begin{remark}[Three characterisations]
\label{rem:two-rigidities}
Conditional on its structural requirements, the logarithmic
coordinate is not a convention but is forced by the group structure
of the composition law.
Together with Proposition~\ref{prop:characterisation} below, the
theory rests on three successive rigidity results:
\begin{equation}\label{eq:two-rigidities}
  \underbrace{\text{multiplicative composition}}_
    {\text{forces } \Phi=-\ln\psi}
  \;\;\xrightarrow{\;\text{then}\;}\;\;
  \underbrace{\text{Princ.~\ref{pr:variance-invariance}--\ref{pr:common-persistence}
    + Prop.~\ref{prop:characterisation} separable score class}}_
    {\text{forces } D=c\,D_{\mathrm{Stein}}}
  \;\;\xrightarrow{\;\text{then}\;}\;\;
  \underbrace{\text{Coupling Identity}}_
    {\text{follows}}.
\end{equation}
Conditional on these structural requirements, neither the
logarithmic coordinate nor the divergence remains arbitrary:
each is uniquely characterised up to positive scale.
\end{remark}

\begin{remark}[Construction from random matrix theory]
\label{rem:RMT}
Under a spiked covariance model
$\Sigma = \Sigma_{\mathrm{ref}}^{1/2}
(I + \sum_i \theta_i v_i v_i^\top)
\Sigma_{\mathrm{ref}}^{1/2}$
with concentration ratio $\gamma=n/T_{\mathrm{obs}}$, the
whitened relative operator is
$A = I + \sum_i \theta_i v_i v_i^\top$
and the Marchenko--Pastur bulk occupies $[\lambda_-,\lambda_+]$ with
$\lambda_\pm=(1\pm\sqrt\gamma)^2$.

\smallskip\noindent
\emph{Bulk eigenvalues} ($\lambda_k\in[\lambda_-,\lambda_+]$):
the sample eigenvector is asymptotically uninformative; the
asymptotic overlap~$\rho_k^2$ vanishes.
Under the calibration $\psi_k:=\rho_k^2$, this gives
$\psi_k\to 0$ and $\Phi_k\to\infty$.

\smallskip\noindent
\emph{Supercritical spikes:}
A population spike of strength $\theta_k$ produces a sample
eigenvalue that separates from the bulk when $\theta_k$ exceeds
the Baik--\mbox{Ben\,Arous}--P\'ech\'e (BBP) critical
threshold~$\theta_{\mathrm{crit}}=\sqrt\gamma$.
Above this threshold the sample eigenvector has non-zero
asymptotic overlap~$\rho_k^2$ with the population eigenvector,
with $\rho_k^2\to 0$ as $\theta_k\downarrow\theta_{\mathrm{crit}}$
and $\rho_k^2\to 1$ as $\theta_k\to\infty$.
A natural asymptotic calibration is $\psi_k:=\rho_k^2$, giving
$\Phi_k=-\ln\rho_k^2$ (reliable for strong spikes).
The precise form of $\rho_k^2(\theta_k,\gamma)$ depends on the
spike model; see~[1] for the eigenvalue transition and~[9] for
the eigenvector overlap asymptotics.

\smallskip\noindent
More generally, any calibrated spectral or stability-based map
into $(0,\infty)$ determines $\psi_k$ and hence~$\Phi_k$ (resampling
stability, cross-validation, factor-model fit, practitioner
judgment), subject to the admissibility
condition~\eqref{eq:admissible}.
SM does not prescribe a particular calibration rule;
it requires only an admissible map $k\mapsto\psi_k\in(0,\infty)$.
In the standard regularisation setting ($\Phi_k\ge 0$,
$\psi_k\le 1$), each mode is attenuated or preserved; the full
signed range $\Phi_k\in\R$ extends the formalism to spectral
amplification.
\end{remark}

\clearpage
\subsection*{The spectral landscape}

\begin{figure}[!ht]
\centering
\includegraphics[width=0.85\textwidth]{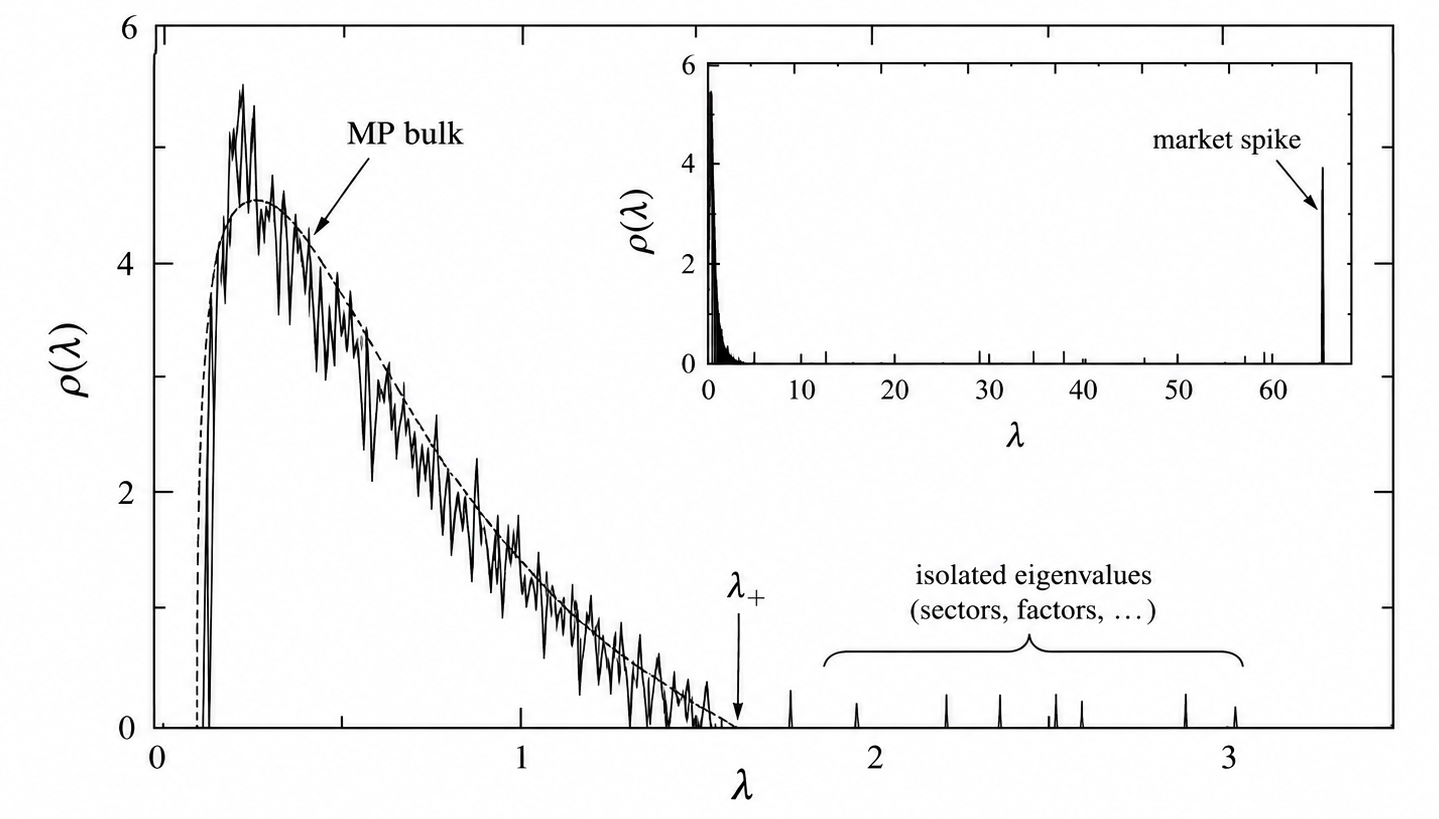}
\caption{Eigenvalue density of a sample correlation matrix in a
spiked covariance model (synthetic data). The dashed line shows the
Marchenko--Pastur density for the bulk. Most eigenvalues lie in the
bulk and correspond to statistically unreliable modes (high
spectral potential, $\Phi_k\to\infty$, $\psi_k\to 0$). A few
isolated eigenvalues lie above the bulk edge~$\lambda_+$,
corresponding to reliable modes (low~$\Phi_k$). The largest
eigenvalue (the ``market spike'') is well separated from the bulk
and shown only in the inset. Inset: same spectrum on a wider scale
including the market spike.
Inspired by an S\&P~500 spectral density plot in Bouchaud and
Potters~[10].}
\label{fig:spectral-landscape}
\end{figure}

\section{The Free Energy}

\subsection{Modelling Principle}

\begin{principle}[Scale-invariant coupling]
\label{pr:variance-invariance}
The regularisation coupling of eigendirection~$k$ depends on its
spectral potential~$\Phi_k$, not on the raw variance
scale~$\lambda_k$.
Formally, the effective coupling $\eta_k\lambda_k$ is determined by
the spectral potential alone:
\begin{equation}\label{eq:modelling-principle}
  \eta_k\lambda_k \;=\; s_k \;=\; e^{\Phi_k}-1
  \qquad\Longrightarrow\qquad
  \eta_k \;=\; \frac{e^{\Phi_k}-1}{\lambda_k}.
\end{equation}
\end{principle}

\begin{principle}[Common persistence]
\label{pr:common-persistence}
The persistence factor~$\psi_k=e^{-\Phi_k}$ governs the
covariance deformation in the same way as it governs the return
signal: the regularised deviation from the reference is
$\psi_k$~times the empirical deviation.  Hence a mode with
persistence~$\psi_k$ is required to satisfy
\begin{equation}\label{eq:common-persistence}
  \lambda_k^* - 1 \;=\; \psi_k\,(\lambda_k - 1).
\end{equation}
For $\psi_k<1$ ($\Phi_k>0$) this is contraction; for $\psi_k>1$
($\Phi_k<0$) it is amplification; for $\psi_k=1$ ($\Phi_k=0$)
it is preservation.
The covariance loss is not specified a priori.
Proposition~\ref{prop:characterisation} shows that, within the
stated loss class and under scale-invariant coupling,
this requirement uniquely characterises Stein's loss.
\end{principle}

\begin{remark}
Principles~\ref{pr:variance-invariance} and~\ref{pr:common-persistence}
separate two logically distinct roles:
$\lambda_k$~measures the empirical variance in direction~$k$;
$\Phi_k$~controls how much of that variance deviation persists
in the regularised state. Making these explicit axioms clarifies which part of
the theory is a modelling choice and which part is a mathematical
consequence.
\end{remark}

\subsection{Structural motivation and scope of the axioms}

The three structural assumptions --- multiplicative persistence
composition (Proposition~\ref{prop:log-char}), scale-invariant
coupling (Principle~\ref{pr:variance-invariance}), and common
persistence (Principle~\ref{pr:common-persistence}) ---
together with the loss-class assumptions of
Proposition~\ref{prop:characterisation}, determine the covariance
loss up to a positive multiplicative constant.
None of the three modelling axioms is derivable from first
principles.  We briefly motivate all three.

\smallskip\noindent
\emph{Multiplicative composition.}
The axiom $\psi_{AB}=\psi_A\psi_B$ models the persistence factor as
a multiplicative filter: when two independently specified
filters act successively, the combined persistence is
the product of the individual factors.
The logarithmic coordinate $\Phi=-\ln\psi$ is then the unique
continuous additive coordinate for this composition law
(Proposition~\ref{prop:log-char}), and the composition extends
to a group: every filter $\psi>0$ has an inverse
$\psi^{-1}=1/\psi$ (equivalently, $\Phi^{-1}=-\Phi$).
The group structure refers to the unconstrained parameter
space~$(0,\infty)$; for a fixed empirical spectrum, admissibility
restricts this group action to the state-dependent
domain~$\mathcal{A}(\lambda)$
(Corollary~\ref{cor:admissible-domain}), which need not be closed
under inversion.
This is a modelling convention for composing spectral
transformations; it does not assert that $\psi$ is a probability
or that empirical estimators must generally compose
multiplicatively.

\smallskip\noindent
\emph{Scale-independent deformation coupling.}
In whitened coordinates, $\lambda_k$ measures relative variance
with respect to the reference value~$1$, whereas $\Phi_k$ controls
the transformation applied to that modal deviation.
The modelling principle requires the spectral potential to determine
the \emph{factor} by which the deviation $\lambda_k-1$ is scaled,
rather than to depend additionally on the magnitude
of~$\lambda_k$.
Thus two modes assigned the same potential undergo the same
proportional transformation, even when their
empirical deviations differ in magnitude.
An eigenvalue-dependent coupling
$\eta_k\lambda_k=s(\Phi_k,\lambda_k)$ would conflate
the spectral potential with deviation size.

\smallskip\noindent
\emph{Common persistence.}
The persistence factor~$\psi_k$ is postulated to act on
covariance deformation in the same way as it acts (via the quadratic
return loss) on returns: the regularised deviation is
$\psi_k$~times the empirical deviation~$\lambda_k - 1$.  This is a
modelling choice about what $\psi_k$~\emph{means}, not a theorem.
The Characterisation Proposition below shows that this requirement
uniquely determines the covariance loss within the stated class.

\smallskip\noindent
Together, the three axioms express a
\emph{separation of concerns}: spectral potential from scale,
deformation coupling from raw magnitude, and a common interpretation of
persistence across returns and covariance.  The Characterisation
Proposition shows that these concerns, once formalised, leave no
remaining freedom in the choice of covariance loss within this
class.

\subsection{Characterisation of the covariance free-energy density}

\begin{proposition}[Characterisation of the covariance loss]
\label{prop:characterisation}
Let $D\in C^2\bigl((0,\infty)^2\bigr)$ satisfy
\begin{equation}\label{eq:separable-deriv}
  \partial_\sigma D(\sigma,\lambda)=a(\lambda)-b(\sigma),
  \qquad a,b\in C^1(0,\infty),
\end{equation}
together with
\begin{equation}\label{eq:D-normalisation}
  D(\lambda,\lambda)=0,\qquad
  D(\sigma,\lambda)>0 \;\text{ for }\; \sigma\neq\lambda.
\end{equation}
For $\lambda>0$ and $s>-1$ with $\lambda+s>0$
(the SM admissibility conditions: $1+s=e^{\Phi}>0$ and
$\lambda+s>0$, equivalently $s>-\min\{1,\lambda\}$),
define the combined loss
\begin{equation}\label{eq:J-def}
  J_{\lambda,s}(\sigma)
  :=
  D(\sigma,\lambda)+\frac{s}{\lambda}\,D(\sigma,1),
\end{equation}
and suppose that $J_{\lambda,s}$ has a unique minimiser
$\sigma^*_{\lambda,s}$ for every admissible $(\lambda,s)$.
Then the following are equivalent:
\begin{enumerate}[nosep,label=(\roman*)]
\item $\displaystyle\sigma^*_{\lambda,s}-1
  =\frac{\lambda-1}{1+s}$
  for all $\lambda>0$, $s>-\min\{1,\lambda\}$.
\item $D(\sigma,\lambda)
  =c\left[\sigma/\lambda-\ln(\sigma/\lambda)-1\right]$
  for some constant $c>0$ (Stein loss).
\end{enumerate}
\end{proposition}

\begin{proof}
$(i)\Rightarrow(ii)$.
Since the hypothesis holds for every admissible
$s>-\min\{1,\lambda\}$, it
holds in particular for $s\ge 0$.  Restricting to this subset is
sufficient for the characterisation; the argument below uses only
$s\ge 0$.

The assumed shrinkage requires
$\sigma^*=(\lambda+s)/(1+s)$ for all $\lambda>0$, $s\ge 0$.
The FOC of~\eqref{eq:J-def}, using~\eqref{eq:separable-deriv},
reads
\[
  a(\lambda)+\frac{s}{\lambda}\,a(1)
  =\left(1+\frac{s}{\lambda}\right)b(\sigma^*).
\]
\emph{Step 1: $a=b$.}
At $s=0$, the combined loss reduces to $J_{\lambda,0}=D(\cdot,\lambda)$.
By~\eqref{eq:D-normalisation}, $D(\lambda,\lambda)=0$ and
$D(\sigma,\lambda)>0$ for $\sigma\neq\lambda$, so $\sigma^*=\lambda$.
The FOC at $s=0$ therefore gives $a(\lambda)-b(\lambda)=0$ for all
$\lambda>0$, i.e.\ $b=a$ on $(0,\infty)$.

\smallskip\noindent
\emph{Step 2: functional equation.}
With $b=a$, the FOC becomes
$a(\lambda)+(s/\lambda)\,a(1)=(1+s/\lambda)\,a(\sigma^*)$.
Substituting $\sigma^*=(\lambda+s)/(1+s)$ yields
\[
  a\!\left(\frac{\lambda+s}{1+s}\right)
  =\frac{\lambda\,a(\lambda)+s\,a(1)}{\lambda+s}
\]
for all $\lambda>0$, $s\ge 0$.
Differentiating in~$s$ at $s=0$ gives the ODE
$a'(\lambda)(1-\lambda)=(a(1)-a(\lambda))/\lambda$.
Continuity of $a$ and $a'$ at $\lambda=1$ identifies the constants
on the two intervals $(0,1)$ and $(1,\infty)$; the general
solution is $a(\lambda)=\alpha/\lambda+\beta$.

\smallskip\noindent
\emph{Step 3: integration.}
The FOC now reads
$\alpha/\lambda-\alpha/\sigma=0$ up to the additive constant.
Integrating $\partial_\sigma D=\alpha(1/\lambda-1/\sigma)$ in
$\sigma$ gives
$D(\sigma,\lambda)=\alpha[\sigma/\lambda-\ln\sigma]+C(\lambda)$.
The normalisation $D(\lambda,\lambda)=0$ determines
$C(\lambda)=-\alpha[1-\ln\lambda]$, yielding
$D=\alpha[\sigma/\lambda-\ln(\sigma/\lambda)-1]$.
Strict positivity $D(\sigma,\lambda)>0$ for $\sigma\neq\lambda$
requires $\alpha>0$; setting $c=\alpha$ completes this direction.

\smallskip\noindent
$(ii)\Rightarrow(i)$.
For Stein loss with $c=1$,
$\partial_\sigma J_{\lambda,s}
=(1+s)/\lambda-(1+s/\lambda)/\sigma$.
Setting this to zero gives
$\sigma^*=(1+s/\lambda)\cdot\lambda/(1+s)=(\lambda+s)/(1+s)$,
hence $\sigma^*-1=(\lambda-1)/(1+s)$.
For admissible~$s>-\min\{1,\lambda\}$, both $1+s>0$ and
$\lambda+s>0$, so $\sigma^*>0$.
\end{proof}

\begin{remark}[Domain asymmetry]
\label{rem:domain-asymmetry}
The two directions of the proof use different domains.
$(i)\Rightarrow(ii)$: the non-negative subset $s\ge 0$ already
determines~$D$ uniquely; the full admissible range
$s>-\min\{1,\lambda\}$ provides no additional information.
$(ii)\Rightarrow(i)$: the Stein-loss identity
$\sigma^*-1=(\lambda-1)/(1+s)$ holds for every
admissible~$s>-\min\{1,\lambda\}$, including the amplifying
regime $s\in(-\min\{1,\lambda\},\,0)$.
Thus the characterisation is proved on $s\ge 0$, but the
\emph{consequence}---the Coupling Identity---extends to the
full signed domain.
\end{remark}

\begin{remark}[Interpretation]
\label{rem:stein-interpretation}
Within this loss class, the Stein loss is not a remaining design
choice: it is the
unique covariance loss whose minimiser under the combined
objective~\eqref{eq:J-def} produces the covariance coupling
identity.
The linear ratio $\sigma/\lambda$ measures relative variance
displacement from the reference eigenvalue, while the logarithmic
term $-\ln(\sigma/\lambda)$ measures log-volume change and provides
the entropic component of the divergence, acting as a barrier
against eigenvalue collapse.
Henceforth we fix the covariance-loss normalisation by setting
$c=1$. This does not affect the equilibrium covariance.
\end{remark}

\subsection{Return shrinkage}

\begin{lemma}[Return coupling]
\label{lem:return-coupling}
For $\lambda>0$, $s>-1$, and $\delta\in\R$, the quadratic return
loss
\begin{equation}\label{eq:return-loss}
  Q_{\lambda,s,\delta}(m)
  :=
  \frac{(m-\delta)^2}{2\lambda}
  +\frac{s}{2\lambda}\,m^2
\end{equation}
has the unique minimiser
\begin{equation}\label{eq:return-shrinkage}
  m^* = \frac{\delta}{1+s} = e^{-\Phi_k}\,\delta.
\end{equation}
\end{lemma}

\begin{proof}
$Q'(m)=(m-\delta)/\lambda+(s/\lambda)\,m=0$ gives
$(1+s)\,m=\delta$.
\end{proof}

\section{The Functional}

\begin{definition}[SM functional]
\label{def:functional}
On the spectral submanifold~$M_A$:
\begin{equation}\label{eq:F-total}
  F = \sum_{k=1}^n F_k(\sigma_k,m_k;\Phi_k),
\end{equation}
where the modal free energy is
\begin{equation}\label{eq:Fk}
  F_k =
  \underbrace{
    D_{\mathrm{Stein}}(\sigma_k,\lambda_k)
    + \frac{1}{2\lambda_k}(m_k-\delta_k)^2
  }_{\text{data fidelity}}
  \;+\;
  \frac{s_k}{\lambda_k}
  \underbrace{
    \left[\,
      D_{\mathrm{Stein}}(\sigma_k,1) + \tfrac{1}{2}m_k^2
    \,\right]
  }_{\text{reference coupling}},
\end{equation}
with $s_k=e^{\Phi_k}-1$ and
$D_{\mathrm{Stein}}(\sigma,\lambda)=\sigma/\lambda-\ln(\sigma/\lambda)-1$.
\end{definition}

\section{Joint Equilibrium}

\begin{theorem}[Coupling Identity]
\label{thm:coupling-identity}
For any admissible potential field
$(\Phi_1,\ldots,\Phi_n)\in\mathcal{A}(\lambda)$, the unique minimiser of~$F$ on~$M_A$
satisfies:
\begin{equation}\label{eq:coupling-identity}
  \boxed{\;
    m_k^* = e^{-\Phi_k}\,\delta_k,
    \qquad
    \lambda_k^*-1 = e^{-\Phi_k}\,(\lambda_k-1).
  \;}
\end{equation}
The Gibbs weight $\psi_k=e^{-\Phi_k}$ governs both observables.
Equivalently:
\begin{equation}\label{eq:coupling-tau}
  m_k^*=(1-\tau_k)\,\delta_k, \qquad
  \lambda_k^*=(1-\tau_k)\,\lambda_k+\tau_k,
\end{equation}
with signed deformation coefficient $\tau_k=1-e^{-\Phi_k}$.
\end{theorem}

\begin{proof}
By admissibility, $\lambda_k+s_k>0$, so
$(1+s_k/\lambda_k)=(\lambda_k+s_k)/\lambda_k>0$.
Each $F_k$ is therefore strictly convex in $(\sigma_k,m_k)$:
the Hessian is
diagonal with entries $(1+s_k/\lambda_k)/\sigma_k^2>0$ (covariance)
and $(1+s_k)/\lambda_k=e^{\Phi_k}/\lambda_k>0$ (return).
The stationary point is the
unique global minimiser.

\smallskip\noindent
\emph{Returns.} $\partial F_k/\partial m_k=0$:
$\lambda_k^{-1}(m_k-\delta_k)+(s_k/\lambda_k)\,m_k=0$.
Multiply by $\lambda_k$: $(m_k-\delta_k)+s_k m_k=0$, hence
$m_k^*=\delta_k/(1+s_k)=e^{-\Phi_k}\delta_k$.

\smallskip\noindent
\emph{Covariance.} $\partial F_k/\partial\sigma_k=0$:
$(1/\lambda_k+s_k/\lambda_k)-(1+s_k/\lambda_k)/\sigma_k=0$,
hence $\sigma_k^*=(\lambda_k+s_k)/(1+s_k)
=e^{-\Phi_k}\lambda_k+(1-e^{-\Phi_k})$.

\smallskip\noindent
\emph{Coupling.} $\sigma_k^*-1=e^{-\Phi_k}(\lambda_k-1)$.
\end{proof}

\begin{remark}[Two legs of the coupling]
\label{rem:equilibrium}
The return leg $m_k^*=e^{-\Phi_k}\delta_k$ follows from the
elementary quadratic structure
(Lemma~\ref{lem:return-coupling}).  The covariance leg
$\lambda_k^*-1=e^{-\Phi_k}(\lambda_k-1)$ is postulated as the
Common Persistence Principle~(\ref{pr:common-persistence});
Proposition~\ref{prop:characterisation} then shows that within the
stated loss class this requirement \emph{forces} the Stein
divergence.  That both legs produce the same Gibbs
weight~$e^{-\Phi_k}$ is the content of the Coupling Identity:
we postulate what the persistence factor is supposed to \emph{mean} for
covariance; the characterisation theorem determines which
divergence can realise that meaning; and the resulting equilibrium
confirms that return and covariance deformation share the same
persistence factor.
\end{remark}

\section{The Pressure Function}

\begin{definition}[SM pressure functional]
\label{def:pressure}
The SM pressure functional is the negative optimal free energy:
\begin{equation}\label{eq:pressure}
  P(\Phi_1,\ldots,\Phi_n)
  := -F^*
  = -\sum_{k=1}^n F_k^*(\Phi_k).
\end{equation}
\end{definition}

\begin{proposition}[Conjugacy relation]
\label{prop:conjugacy}
By the envelope theorem:
\begin{equation}\label{eq:conjugacy}
  \boxed{\;
    \frac{\partial P_k}{\partial\Phi_k}
    = -\frac{e^{\Phi_k}}{\lambda_k}\,D^{\mathrm{ref}}_{k,\mathrm{opt}},
  \;}
\end{equation}
where $D^{\mathrm{ref}}_{k,\mathrm{opt}}=D_{\mathrm{Stein}}(\lambda_k^*,1)+\frac{1}{2}(m_k^*)^2$
is the residual reference divergence --- how far the regularised
mode remains from the reference state.
\end{proposition}

\begin{proof}
The only $\Phi_k$-dependence in $F_k$ enters through
$s_k=e^{\Phi_k}-1$. Since $\partial F_k/\partial\sigma_k=0$ and
$\partial F_k/\partial m_k=0$ at the optimum,
$dF_k^*/d\Phi_k
=(\partial F_k/\partial s_k)\cdot(ds_k/d\Phi_k)
=(D^{\mathrm{ref}}_{k,\mathrm{opt}}/\lambda_k)\cdot e^{\Phi_k}$.
\end{proof}

\begin{remark}[Interpretation]
\label{rem:conjugacy-interpretation}
This relation justifies treating~$\Phi_k$ as a genuine control
parameter of the theory: it has a well-defined conjugate response.
The structure $(\Phi_k,\,D^{\mathrm{ref}}_{k,\mathrm{opt}})$ is the spectral-potential analogue
of a conjugate pair in the variational formalism.

The term \emph{pressure} in SM denotes the negative value function
of the variational problem ($P_k(0)=0$, $P_k(\Phi_k)<0$ for
non-trivial modes at $\Phi_k>0$).  It is used in deliberate analogy
with thermodynamic formalism; its absolute sign carries no
thermodynamic interpretation.

The conjugacy relation is intrinsic to the SM variational problem.
An exact embedding into Ruelle's thermodynamic formalism would
additionally require the identification of an underlying dynamical
system, invariant measure, entropy functional, and potential for
which the SM pressure satisfies the corresponding variational
principle.

\smallskip\noindent
\emph{Limits.}
As $\Phi_k\to\infty$:
$D^{\mathrm{ref}}_{k,\mathrm{opt}}=O(e^{-2\Phi_k})$, so
$e^{\Phi_k}D^{\mathrm{ref}}_{k,\mathrm{opt}}=O(e^{-\Phi_k})\to 0$ and
$\partial P_k/\partial\Phi_k\to 0$.
As $\Phi_k\to 0$: the mode retains all data information,
$D^{\mathrm{ref}}_{k,\mathrm{opt}}\to D_{\mathrm{Stein}}(\lambda_k,1)+\delta_k^2/2$, the full
data--reference divergence.
\end{remark}

\begin{proposition}[Monotonicity, convexity, and analyticity of
the modal pressure]
\label{prop:convexity}
For any non-trivial mode $(\lambda_k,\delta_k)\neq(1,0)$ and all
admissible~$\Phi_k$ (i.e.\
$\Phi_k>\ln(1-\lambda_k)$ if $\lambda_k<1$,
$\Phi_k\in\R$ if $\lambda_k\ge 1$):
\begin{enumerate}[nosep,label=(\roman*)]
\item \emph{Closed form.}
\begin{equation}\label{eq:Fstar}
  F_k^*(\Phi_k) =
  \frac{\lambda_k+e^{\Phi_k}-1}{\lambda_k}
  \ln\frac{e^{\Phi_k}}{\lambda_k+e^{\Phi_k}-1}
  +\ln\lambda_k
  +\frac{(1-e^{-\Phi_k})\delta_k^2}{2\lambda_k}.
\end{equation}
\item \emph{Boundary values.}
$P_k(0)=0$ and
\begin{equation}\label{eq:P-limit}
  P_k(\Phi_k)\;\to\;
  -\!\left[D_{\mathrm{Stein}}(1,\lambda_k)+\frac{\delta_k^2}{2\lambda_k}\right]
  \quad\text{as }\Phi_k\to\infty.
\end{equation}
\item \emph{Strict monotonicity.}
For every admissible~$\Phi_k$, $P_k'(\Phi_k)<0$ for a non-trivial
mode, while $P_k'(\Phi_k)\to 0^-$ as $\Phi_k\to\infty$.
\item \emph{Strict convexity.}
The modal susceptibility is
\begin{equation}\label{eq:susceptibility}
  \chi_k(\Phi_k)
  :=P_k''(\Phi_k)
  =\frac{e^{\Phi_k}}{\lambda_k}
  \!\left[\,
    \ln\lambda_k^*+\frac{1}{\lambda_k^*}-1
    +\frac{e^{-2\Phi_k}\delta_k^2}{2}
  \,\right]
  > 0,
\end{equation}
where $\lambda_k^*=1+e^{-\Phi_k}(\lambda_k-1)$.
\item \emph{Analyticity.}
$P_k(\Phi_k)$ is real-analytic on the entire admissible interval,
for every $\lambda_k>0$.
\end{enumerate}
\end{proposition}

\begin{proof}
\emph{(i)}~follows by substituting $(\sigma_k^*,m_k^*)$ from
Theorem~\ref{thm:coupling-identity} into~\eqref{eq:Fk} and
simplifying. \emph{(ii)}~follows from (i) at $\Phi_k=0$ and
$\Phi_k\to\infty$, where $\lambda_k^*\to 1$, $m_k^*\to 0$.
\emph{(iii)}~is Proposition~\ref{prop:conjugacy}, since
$D^{\mathrm{ref}}_{k,\mathrm{opt}}\ge 0$. For \emph{(iv)}, set $x=\lambda_k^*$. Then
$dx/d\Phi_k=-(x-1)$, and
\[
  \frac{d}{d\Phi_k}D_{\mathrm{Stein}}(x,1)
  =\left(1-\frac{1}{x}\right)\!\cdot\!(-(x-1))
  =-\frac{(x-1)^2}{x}.
\]
Combined with $(d/d\Phi_k)(e^{-2\Phi_k}\delta_k^2/2)
=-e^{-2\Phi_k}\delta_k^2$, this gives
$P_k''=-(e^{\Phi_k}/\lambda_k)[D^{\mathrm{ref}}_{k,\mathrm{opt}}+(d/d\Phi_k)D^{\mathrm{ref}}_{k,\mathrm{opt}}]$.
The bracket reduces to $D_{\mathrm{Stein}}(x,1)-(x-1)^2/x
-e^{-2\Phi_k}\delta_k^2/2$.
But
\[
  \frac{(x-1)^2}{x}-D_{\mathrm{Stein}}(x,1)
  =\frac{(x-1)^2}{x}-x+\ln x+1
  =\ln x+\frac{1}{x}-1,
\]
and the function $g(x)=\ln x+1/x-1$ satisfies $g'(x)=(x-1)/x^2$,
so $g$ has its unique global minimum at $x=1$ with $g(1)=0$.
Hence
$P_k''=(e^{\Phi_k}/\lambda_k)[\ln x+1/x-1
+e^{-2\Phi_k}\delta_k^2/2]>0$.
\emph{(v)}~follows because the closed
form~\eqref{eq:Fstar} is composed of exponentials and logarithms
of strictly positive analytic functions of~$\Phi_k$.
\end{proof}

\begin{remark}[No intrinsic pressure singularities in SM]
\label{rem:no-phase-transition}
Since each $P_k$ is real-analytic and the finite sum
$P=\sum_k P_k$ preserves analyticity, finite-dimensional SM has no
intrinsic pressure singularities.
A non-analyticity can nonetheless arise through the spectral
\emph{calibration} $\lambda\mapsto\psi(\lambda)\mapsto\Phi_k(\lambda)$
--- for instance a hard Bulk/Spike rule at the BBP threshold ---
even though $P(\Phi)$ itself is analytic. This is a
calibration-induced non-analyticity (and, for a hard threshold,
possibly a discontinuity), not an intrinsic phase transition of the
pressure functional.
Singular behaviour intrinsic to the pressure, if present in
General Markowitz, would require additional structure beyond
the SM assumptions --- for example mode interactions together
with loss of uniqueness or an appropriate large-system limit.
\end{remark}

\begin{proposition}[Additivity]
\label{prop:additivity}
In SM the pressure is additive across modes:
$P=\sum_k P_k(\Phi_k)$. The cross-susceptibility vanishes:
\begin{equation}\label{eq:cross-zero}
  \frac{\partial^2 P}{\partial\Phi_j\,\partial\Phi_k}=0,
  \quad j\neq k.
\end{equation}
The modes are variationally decoupled. This is the defining property of
Special Markowitz (Definition~\ref{def:SM}).
\end{proposition}

\begin{remark}[General Markowitz]
\label{rem:GM}
When the eigendirections of~$A$ are not stable (non-Gaussian
returns, time-varying correlations), the spectral decomposition no
longer diagonalises the problem. Instability of the eigenbasis may
induce interaction terms $J_{jk}V_{jk}$ between modes, breaking the
additive structure:
$\partial^2 P/\partial\Phi_j\partial\Phi_k\neq 0$.
General Markowitz is defined by the presence of such interactions.
\end{remark}

\section{Special Markowitz}

\begin{definition}[Special Markowitz]
\label{def:SM}
\textbf{Special Markowitz} is defined by separability of the free
energy: admissible spectral potentials act independently across the
eigendirections of the whitened relative operator,
\begin{equation}\label{eq:SM-def}
  F = \sum_k F_k.
\end{equation}
Consequently, the optimal pressure is additive,
$P=\sum_k P_k$,
and all cross-susceptibilities vanish:
$\partial^2 P/\partial\Phi_j\partial\Phi_k=0$ for $j\neq k$.
\end{definition}

\section{Properties}

\begin{proposition}
\label{prop:properties}
For any admissible potential field:
\begin{enumerate}[nosep,label=(\roman*)]
\item \emph{Neutral modes: signal preservation.}
For $\Phi_k= 0$: $\psi_k= 1$. Return and eigenvalue
deviation are preserved exactly.
\item \emph{Attenuating modes ($\Phi_k>0$): noise suppression.}
$\psi_k=e^{-\Phi_k}< 1$. Return and
eigenvalue deviation are contracted toward the reference.
As $\Phi_k\to\infty$: $\lambda_k^*\to 1$, $m_k^*\to 0$.
\item \emph{Amplifying modes ($\Phi_k<0$): signal enhancement.}
$\psi_k=e^{-\Phi_k}> 1$. The empirical
deviation $\lambda_k-1$ is amplified beyond its data value,
subject to admissibility~\eqref{eq:admissible}.
\item \emph{Deformation scaling.}
$|\lambda_k^*-1|=\psi_k\,|\lambda_k-1|$,
$|m_k^*|=\psi_k\,|\delta_k|$.
The persistence factor controls the magnitude of the deviation
from the reference, in both covariance and return.
\item \emph{Monotonicity.}
$\psi_k$ is strictly decreasing in~$\Phi_k$.  Increasing the
potential always moves the mode toward the reference.
\end{enumerate}
\end{proposition}

\section{The Relative Markowitz Problem}

SM operates relative to a reference state. The natural portfolio
problem is the relative Markowitz problem in whitened spectral
coordinates: maximise the regularised excess return over the
reference, penalised by regularised risk:
\begin{equation}\label{eq:relative-markowitz}
  w^{SM}
  := \arg\max_w
  \left\{w^\top m^*-\tfrac{1}{2}w^\top\Lambda^* w\right\}
  = (\Lambda^*)^{-1}m^*.
\end{equation}
Modally: $w_k^{SM}=m_k^*/\lambda_k^*$.

\begin{proposition}[Three equivalent views]
\label{prop:three-views}
\begin{equation}\label{eq:wk}
  w_k^{SM}
  = \frac{\delta_k}{\lambda_k+e^{\Phi_k}-1}
  = \frac{\lambda_k}{\lambda_k+e^{\Phi_k}-1}\cdot
    \frac{\delta_k}{\lambda_k}
  = R_k\cdot w_k^{\mathrm{Markowitz}}.
\end{equation}
\end{proposition}

\begin{proof}
$w_k=e^{-\Phi_k}\delta_k/(e^{-\Phi_k}\lambda_k+1-e^{-\Phi_k})
=\delta_k/(\lambda_k+e^{\Phi_k}-1)$.
\end{proof}

\begin{remark}[Three interpretations]
\leavevmode
\begin{enumerate}[nosep,label=(\roman*)]
\item \emph{Spectral ridge.}
$w_k=\delta_k/(\lambda_k+s_k)$: Markowitz with a
potential-dependent ridge $s_k=e^{\Phi_k}-1$.
\item \emph{Spectral gate.}
$w_k=R_k\,w_k^{\mathrm{Markowitz}}$ with gate
$R_k=\lambda_k/(\lambda_k+s_k)$.
For $\Phi_k\ge 0$: $R_k\in[0,1]$ (attenuation).
For $\Phi_k<0$: $R_k>1$ (amplification beyond Markowitz).
\item \emph{Joint deformation.}
The same $\psi_k=e^{-\Phi_k}$ scales the return and the covariance
deviation from the reference.
\end{enumerate}
\end{remark}

\section{The Uniform Limit}

\begin{corollary}
\label{cor:uniform}
For $\Phi_k=\Phi$ for all~$k$ (uniform spectral potential):
$\psi_k=e^{-\Phi}$, and
\begin{align}
  m_k^* &= e^{-\Phi}\,\delta_k, \label{eq:uniform-m}\\
  \lambda_k^* &= e^{-\Phi}\,\lambda_k+(1-e^{-\Phi}).
  \label{eq:uniform-lam}
\end{align}
The portfolio weight becomes
$w_k=\delta_k/(\lambda_k+e^{\Phi}-1)$ --- a constant ridge.
This is the sanity check: without spectral discrimination,
SM reduces to isotropic regularisation.
\end{corollary}

\section{Connection to Existing Methods}
\label{sec:connections}

\begin{proposition}
\label{prop:connections}
\leavevmode
\begin{enumerate}[nosep,label=(\roman*)]
\item \emph{Ledoit--Wolf.}
Linear Ledoit--Wolf covariance shrinkage is obtained as the
uniform-potential, covariance-only limit of SM
(Section~\ref{sec:lw}).
\item \emph{Black--Litterman-like limiting case.}
When $A=I$ and $\Phi_k=\Phi$:
$\lambda_k^*=1$, $m_k^*=e^{-\Phi}\delta_k$ --- isotropic return
shrinkage.
\item \emph{Ridge regression.}
Markowitz with a constant ridge~$\alpha$:
$w_k=\delta_k/(\lambda_k+\alpha)$.
SM replaces the constant by a potential-dependent
$s_k=e^{\Phi_k}-1$.
\end{enumerate}
\end{proposition}

\subsection{Ledoit--Wolf Shrinkage}
\label{sec:lw}

\begin{proposition}[Linear Ledoit--Wolf as uniform-potential SM]
\label{prop:lw}
Let $\tau\in[0,1)$ and set $\Phi_k=\Phi:=-\ln(1-\tau)$ for all~$k$.
Then $\psi=e^{-\Phi}=1-\tau$ and, in whitened coordinates,
\begin{equation}\label{eq:lw-whitened}
  A^* = e^{-\Phi}\,A + (1-e^{-\Phi})\,I
      = (1-\tau)\,A + \tau\,I.
\end{equation}
Transforming back to the original coordinates gives
\begin{equation}\label{eq:lw-original}
  \boxed{\;
    \Sigma^* = (1-\tau)\,\hat\Sigma + \tau\,\Sigma_{\mathrm{ref}}.
  \;}
\end{equation}
For the classical choice
$\Sigma_{\mathrm{ref}}=\mu\,I$~{\upshape[3]},
\eqref{eq:lw-original} has exactly the linear Ledoit--Wolf
shrinkage form; Ledoit and Wolf additionally provide a statistical
estimator of the optimal shrinkage intensity~$\tau$.
\end{proposition}

\begin{proof}
In whitened coordinates ($\Sigma_{\mathrm{ref},w}=I$),
Corollary~\ref{cor:uniform} gives
$\lambda_k^*=(1-\tau)\lambda_k+\tau$ for every~$k$, hence
$\Lambda^*=(1-\tau)\Lambda+\tau\,I$ and
$A^*=U\Lambda^*U^\top=(1-\tau)\,A+\tau\,I$.
Reversing the whitening
$A^*=\Sigma_{\mathrm{ref}}^{-1/2}\Sigma^*\Sigma_{\mathrm{ref}}^{-1/2}$
yields~\eqref{eq:lw-original}.
\end{proof}

\paragraph{What SM adds.}
SM extends linear Ledoit--Wolf in three directions:
\begin{enumerate}[nosep,leftmargin=1.5em]
\item \emph{Mode-dependent potentials.}
  The uniform $\Phi$ is replaced by a spectral field
  $(\Phi_1,\ldots,\Phi_n)$, so that each eigendirection of the
  whitened relative operator receives its own deformation factor.
\item \emph{Signed potentials.}
  $\Phi_k<0$ is admissible: individual modes may be amplified
  beyond their empirical deviation, not only attenuated.
\item \emph{Joint return--covariance deformation.}
  LW produces a covariance estimator.
  SM couples the same persistence factor~$\psi_k=e^{-\Phi_k}$
  to both covariance deviation and expected return
  (Theorem~\ref{thm:coupling-identity}).
  This joint coupling distinguishes SM from covariance-only
  shrinkage estimators such as Ledoit--Wolf.
\end{enumerate}

\subsubsection*{Nonlinear Ledoit--Wolf}

Nonlinear Ledoit--Wolf shrinkage~[4] replaces the uniform shrinkage
intensity~$\tau$ by eigenvalue-dependent corrections, producing
target eigenvalues~$d_k$ that depend nonlinearly on the sample
spectrum.  For modes with $\lambda_k=1$, SM produces
$\lambda_k^*=1$ for every finite~$\Phi_k$, so no non-trivial
target arises.  For $\lambda_k\neq 1$, if such a target is
representable within SM, the representing potential is
\begin{equation}\label{eq:nlw-phi}
  \Phi_k^{\mathrm{NLW}}
  = -\ln\!\left(\frac{d_k-1}{\lambda_k-1}\right),
\end{equation}
which requires
$\operatorname{sign}(d_k-1)=\operatorname{sign}(\lambda_k-1)$
and $d_k>0$.  Whether the nonlinear Ledoit--Wolf target eigenvalues
satisfy these conditions for all spectral configurations---in
particular near the reference crossing $\lambda_k=1$---is a
quantitative question left for a separate treatment.

\subsection{Agnostic Risk Parity}
\label{sec:arp}

Benichou et~al.~[8] propose \emph{Agnostic Risk
Parity} (ARP), a portfolio construction that replaces the Markowitz
operator~$C^{-1}$ by~$C^{-1/2}$.  The structural kinship with
Special Markowitz is instructive; so is the precise point at which
the two constructions diverge.

\paragraph{Shared spectral geometry.}
ARP begins by whitening the risk space via
$\hat{r} = C^{-1/2}r$, rendering the covariance in the
transformed space isotropic.  SM performs a closely related
operation, but \emph{relative to a reference state}:
$A = \Sigma_{\mathrm{ref}}^{-1/2}\,\hat\Sigma\,
    \Sigma_{\mathrm{ref}}^{-1/2}
= U\Lambda U^\top$.
Both frameworks exploit the spectral geometry of a whitened
covariance operator, although for different purposes.

\paragraph{Divergence.}
ARP asks: \emph{If we know nothing reliable about preferred
directions, which portfolio map respects the resulting rotational
symmetry?}  The answer, under an agnostic assumption on the
predictor covariance ($Q \propto I$), is
$w_k^{\mathrm{ARP}}
\propto p_k/\sqrt{\lambda_k}$,
where $p_k$ is the projection of the return indicator onto the
$k$-th eigenmode.  Markowitz's
$\lambda_k^{-1}$ is softened to~$\lambda_k^{-1/2}$, equalising
the expected risk contribution across principal components.

SM asks a different question: \emph{If we possess heterogeneous
statistical evidence across modes, how much of the empirical
deviation from the reference should survive in each direction?}
The answer is the spectral gate~\eqref{eq:wk}.

\paragraph{The ARP spectral filter as an SM-representable potential.}
On the \emph{portfolio level}, the SM framework can reproduce the
ARP spectral filter by a particular choice of~$\Phi$.  (This does
not make the complete ARP theory a special case of SM; see the
differences below.)  Setting
$p_k = \delta_k$ and requiring
$w_k^{\mathrm{SM}} \propto w_k^{\mathrm{ARP}}$, we need
$(\lambda_k + e^{\Phi_k} - 1)^{-1}
= c/\sqrt{\lambda_k}$
for some global constant~$c > 0$, hence
\begin{equation}\label{eq:Phi-ARP}
\Phi_k^{\mathrm{ARP}}
= \ln\!\Bigl(1 + \tfrac{\sqrt{\lambda_k}}{c}
  - \lambda_k\Bigr).
\end{equation}
For a real-valued potential, \eqref{eq:Phi-ARP} requires only that
its argument be positive, i.e.\ $\sqrt{\lambda_k}/c > \lambda_k-1$.
No sign restriction on~$\Phi_k$ is imposed a priori.  The condition
$\Phi_k^{\mathrm{ARP}}\ge 0$ (no modal amplification beyond
Markowitz) holds if and only if $c\le 1/\sqrt{\lambda_k}$;
a sufficient condition is $c \le 1/\sqrt{\lambda_{\max}}$.

More generally, any positive spectral filter $g(\lambda)$ satisfying
$w_k = g(\lambda_k)\,\delta_k$ can be expressed via
\begin{equation}\label{eq:Phi-general}
\Phi(\lambda) = \ln\!\Bigl(1 + \frac{1}{g(\lambda)} - \lambda\Bigr)
\end{equation}
whenever the argument is positive, i.e.\ $1/g(\lambda) > \lambda-1$
(admissibility).  Filters more conservative than Markowitz
($1/g(\lambda)\ge\lambda$) produce $\Phi\ge 0$; filters more
aggressive than Markowitz ($1/g(\lambda)<\lambda$) produce
$\Phi<0$.  The boundary $\Phi_k=0$ corresponds to Markowitz
itself: $g(\lambda)=\lambda^{-1}$.

This makes~$\Phi$ a
\textbf{common language for spectral portfolio transformations}:
Markowitz corresponds to $\Phi_k = 0$
($g = \lambda^{-1}$, the identity);
ARP to $\Phi_k = \Phi^{\mathrm{ARP}}(\lambda_k)$
($g = c/\!\sqrt{\lambda}$);
MP reliability to
$\Phi_k = \Phi^{\mathrm{MP}}(\lambda_k;\,n/T)$;
and aggressive filters ($g>\lambda^{-1}$ in some modes) to
$\Phi_k<0$ in those modes.

\paragraph{Three substantive differences.}
Despite this representability on the portfolio level, ARP and~SM
differ in structure:
\begin{enumerate}[nosep,leftmargin=1.5em]
\item \textbf{Symmetry vs.\ heterogeneous spectral potential.}\enspace
ARP imposes rotational invariance by treating all predictor
directions as statistically equivalent ($Q \propto I$).  SM
permits---and is designed for---mode-dependent spectral potentials
through the field~$(\Phi_k)$.

\item \textbf{No joint regularisation in ARP.}\enspace
ARP regularises the risk operator ($C^{-1} \to C^{-1/2}$) and
treats return indicators separately via their
covariance~$Q$.  SM performs a \emph{joint} scaling of
returns and covariance through a single persistence
factor~$\psi_k=e^{-\Phi_k}$ per mode
(Theorem~\ref{thm:coupling-identity}).  This joint coupling has
no counterpart in ARP.

\item \textbf{Variational origin.}\enspace
ARP is derived from a symmetry argument; no explicit optimisation
is involved.  SM minimises a spectral free energy and
\emph{characterises} the Stein loss as the unique covariance
divergence compatible with the Coupling Identity.
\end{enumerate}

\noindent
ARP can be recovered \emph{at the portfolio level} as a particular
spectral potential within the SM framework~\eqref{eq:Phi-ARP}.
The SM formalism does not subsume ARP's symmetry-based theory;
rather, it provides a variational \emph{representation} of the ARP
spectral filter within a unified free-energy framework.
Since ARP defines portfolio weights only up to global scaling,
the constant~$c$ is not determined by ARP alone: different
normalisations produce a one-parameter family of representing
SM potentials.  Fixing a portfolio normalisation fixes~$c$, and
hence the representing potential.

\section{Outlook}

SM is the decoupled, additive-pressure theory: variationally decoupled modes,
static potential field, analytic pressure functional. Three
generalisations follow from relaxing these conditions.

\medskip
\begin{center}
\begin{tabular}{@{}lll@{}}
\toprule
& \textbf{Structure} & \textbf{Consequence} \\
\midrule
SM & $F=\sum_k F_k$ & modes decouple,
  $\partial^2 P/\partial\Phi_j\partial\Phi_k=0$ \\[3pt]
GM & $F=\sum_k F_k+\sum_{j\neq k}J_{jk}V_{jk}+\cdots$
  & modes couple,
  $\partial^2 P/\partial\Phi_j\partial\Phi_k\neq 0$ \\[3pt]
Dynamic & $\Phi_k(t)$
  & time-dependent potential, flow equations \\
\bottomrule
\end{tabular}
\end{center}

\medskip\noindent
The SM/GM boundary is sharp: a future General Markowitz theory
would be characterised by the
presence of interaction terms $J_{jk}$ that break the separable
structure. These may arise when the eigendirections of~$A$ are not
stable --- non-Gaussian returns, time-varying correlations, or
non-commuting data and noise operators.

In a synergetic analogy: the eigendirections above the
Marchenko--Pastur threshold play the role of order parameters (reliable,
slow, persistent); the bulk directions play the role of enslaved modes
(unreliable, fast, noise). SM suppresses the enslaved modes toward
the reference and lets the order parameters determine the portfolio.

\begin{remark}[Phase transitions --- a programme]
\label{rem:phase}
When an eigenvalue crosses the MP threshold, an order parameter
appears or disappears and the dimension of the effective portfolio
manifold changes. By Proposition~\ref{prop:convexity}, the SM
pressure $P(\Phi)$ is itself real-analytic --- no intrinsic
singularity occurs. However, the spectral calibration
$\lambda\mapsto\Phi_k(\lambda)$ (Remark~\ref{rem:RMT}) can
introduce a non-smooth dependence on the data eigenvalues at the
BBP threshold, producing a calibration-induced non-analyticity
in $P(\Phi(\lambda))$ as a function of~$\lambda$
(Remark~\ref{rem:no-phase-transition}). Intrinsic pressure
singularities are excluded in finite-dimensional SM by the
analyticity of the modal pressures and their finite additive
composition (Proposition~\ref{prop:convexity}).
Whether the non-separable interactions of General Markowitz,
together with additional mechanisms such as loss of uniqueness
or a large-system limit, can produce genuine phase transitions
is an open question deferred to subsequent work.
\end{remark}

\medskip\noindent
Special Markowitz is the separable, additive-pressure theory.
General Markowitz is its interacting, non-additive extension.

\bigskip\noindent
\paragraph{Logical architecture.}
The theory proceeds in the following chain:
\begin{equation}\label{eq:chain}
  \underbrace{\psi_k\in(0,\infty)\vphantom{\Big|}}_
    {\text{persistence factor}}
  \;\xrightarrow{\;-\ln\;}
  \;\underbrace{\Phi_k\in\R\vphantom{\Big|}}_
    {\text{signed spectral potential}}
  \;\xrightarrow{\;e^{\Phi}\!-1\;}
  \;\underbrace{s_k\vphantom{\Big|}}_
    {\text{coupling strength}}
  \;\longrightarrow\;
  \underbrace{(m_k^*,\,\lambda_k^*\!-\!1)
    =\psi_k(\delta_k,\,\lambda_k\!-\!1)
    \vphantom{\Big|}}_
    {\text{joint regularisation}}
  \;\longrightarrow\;
  \underbrace{w_k^{SM}\vphantom{\Big|}}_
    {\text{portfolio}}.
\end{equation}
The variational formalism provides the mathematical structure ---
potential, Gibbs weight, free energy, pressure, conjugacy, and
separability. An exact representation of the SM pressure within
Ruelle's thermodynamic formalism remains an open mathematical
question.

\section*{References}

\begin{enumerate}[label={[\arabic*]},nosep,leftmargin=2em]
\item J.\ Baik, G.\ Ben Arous, S.\ P\'ech\'e,
``Phase transition of the largest eigenvalue,''
\textit{Ann.\ Probab.}\ \textbf{33}, 1643 (2005).
\item F.\ Black, R.\ Litterman,
``Global portfolio optimization,''
\textit{FAJ}\ \textbf{48}(5), 28 (1992).
\item O.\ Ledoit, M.\ Wolf,
``A well-conditioned estimator,''
\textit{JMVA}\ \textbf{88}, 365 (2004).
\item O.\ Ledoit, M.\ Wolf,
``Analytical nonlinear shrinkage,''
\textit{Ann.\ Statist.}\ \textbf{48}, 3043 (2020).
\item D.\ Ruelle,
\textit{Thermodynamic Formalism}, 2nd ed.,
Cambridge University Press (2004).
\item R.\ Bowen,
\textit{Equilibrium States and the Ergodic Theory of
Anosov Diffeomorphisms}, Lecture Notes in Mathematics~470,
Springer (1975).
\label{Bowen1975}
\item G.\ Bianconi,
``Gravity from entropy,''
\textit{Phys.\ Rev.\ D}\ \textbf{111}, 066001 (2025).
\label{Bianconi2025}
\item R.\ Benichou, Y.\ Lemperiere, E.\ Serie, J.\ Kockelkoren,
P.\ Seager, J.-P.\ Bouchaud, and M.\ Potters,
``Agnostic Risk Parity: Taming Known and Unknown-Unknowns,''
\textit{J.\ Investment Strategies}\ \textbf{6}(3), 2017;
arXiv:1610.08818.
\label{Benichou2016}
\item F.\ Benaych-Georges, R.\ R.\ Nadakuditi,
``The eigenvalues and eigenvectors of finite, low rank
perturbations of large random matrices,''
\textit{Adv.\ Math.}\ \textbf{227}, 494 (2011).
\label{BenGeo2011}
\item J.-P.\ Bouchaud, M.\ Potters,
\textit{Theory of Financial Risk and Derivative Pricing},
2nd ed., Cambridge University Press (2003).
\label{BouchaudPotters2003}
\end{enumerate}

\end{document}